\documentclass{article}
\usepackage{amsmath}
\usepackage{amssymb}
\usepackage{hyperref}

\title{Invisible pushdown languages}

\author{Eryk~Kopczy\'nski}

\begin{document}

\maketitle

\newcounter{mycount}[section]
\renewcommand{\themycount}{\thesection.\arabic{mycount}}

\newtheorem{theorem}[mycount]{Theorem}
\newtheorem{lemma}[mycount]{Lemma}
\newtheorem{definition}[mycount]{Definition}
\newtheorem{problem}[mycount]{Problem}

\newenvironment{proof}[1][]%
{\medskip{\bf Proof #1 }}%
{}
\def\qed{\hfill$\rule{2mm}{2mm}$\par\medskip}

\def\mod{{\rm{\ mod\ }}}
\def\ra{\rightarrow}
\def\la{\leftarrow}
\def\trans{\delta}
\def\rh{|}
\def\ut{\tau}
\def\emt{\emptyset}
\def\Port{{\rm{port}}}
\def\Aut{\mathcal A}
\def\Lang{\mathcal L}
\def\restrict{{\rm{Restrict}}}

\def\dom{{\rm{dom}}}
\def\bia{\{0,1\}}
\def\i{\hskip 2em}
\def\EOT{{\rm EOT}}
\def\pat{p}

\def\calC{{\mathcal C}}
\def\bbN{{\mathbb N}}
\def\PaV{V^P}
\def\PaE{E^P}
\def\SpV{{V^\tau}}
\def\SpE{{E^\tau}}
\def\SpR{{v^\tau}}

\def\bbN{{\mathbb N}}
\def\bbR{{\mathbb R}}
\def\bbP{{\mathbb P}}
\def\bbZ{{\mathbb Z}}
\def\bbC{{\mathbb C}}
\def\oa{\diamondsuit}
\def\ca{\overbar{\diamondsuit}}
\def\leaf{\spadesuit}

\newcommand{\overbar}[1]{\mkern 1.5mu\overline{\mkern-1.5mu#1\mkern-1.5mu}\mkern 1.5mu}

\begin{abstract}
Context free languages allow one to express data with hierarchical structure,
at the cost of losing some of the useful properties of languages recognized by 
finite automata on words. However, it is possible
to restore some of these properties by making the structure of the tree visible, 
such as is done by visibly pushdown languages, or finite automata on trees.
In this paper, we show that the structure given by such approaches remains
invisible when it is read by a finite automaton (on word).
In particular, we show
that separability with a regular language
is undecidable for visibly pushdown languages, just as it is undecidable for
general context free languages.
\end{abstract}

\section{Introduction}

Finite automata are a well known formalism for describing the simplest formal languages.
Regular languages -- ones which are recognized by finite automata -- have very nice
closure properties, such as decidability of most problems such as universality or
disjointness, equivalence of deterministic finite automata (DFA) and non-deterministic
finite automata (NFA), and closure under complement.

However, most programming and natural languages have to describe a hierarchical (tree)
structure, and finite automata on words are no longer appropriate.
To capture such a hierarchical structure,
Noam Chomsky proposed the classic notion of \emph{context free languages}.
Context free languages are recognized
by context free grammars (CFGs), or equivalently by pushdown automata (PDA).

However, context free languages do not have as good properties as regular ones --
for example, universality and disjointness are no longer decidable, deterministic
PDA are less powerful than non-deterministic ones, and they are not closed under
complement. These properties fail since, although
words from a context free language have an underlying tree structure, it is hard
to tell what this structure is just by looking at the word -- two completely different
derivation trees can yield a very similar output, consider for example the English
sentences \emph{Time flies like an arrow} and \emph{fruit flies like a banana}, 
or \emph{The complex houses married and single soldiers and their families} -- after
reading the four first words of the latter sentence, one could think that 
\emph{the complex houses} is the subject and \emph{married} is the verb, while in fact,
\emph{the complex} is the subject and \emph{houses} is the verb. This is also a big
problem in practical computer science, since such a possibility of incorrect parsing leads to many errors
-- one famous example is the SQL injection attack, which is based on fabricating
SQL queries which will be parsed incorrectly, allowing unauthorized access to a database.

There are two popular approaches to solve this. One of them is to use the 
\emph{finite automata on trees} (TFAs) \cite{tata}, which work on trees directly. Another one is to
use \emph{visibly pushdown automata} (VPDAs), also known as languages of \emph{nested
words} \cite{visibly}, where every symbol in our alphabet has a fixed type
with respect to the stack -- it either always pushes a new symbols, or always pops a
symbol, or it never pushes or pops symbols -- this property allows the tree structure to be easily
read. When we safely flatten a language on trees into a language of words -- XML is the
common and effective way to do that -- these two approaches are seen to be equivalent
(in some sense), and most properties of regular languages of words are retained --
non-deterministic and deterministic VPDAs and finite automata on trees are equivalent,
and universality and intersection problems are decidable. Hence,
representing our data as trees, instead of forcing a linear word structure,
definitely solves many problems -- both theoretical and practical -- efficiently.

In this paper, we show that not all problems are solved by these approaches. In
particular, we show that, informally,
although (flattened) TFAs and VPDAs are successful at making
the structure visible to powerful computation models such as
Turing machines, the structure still remains invisible to the simple
ones, such as finite automata on words. We use our technique to show that the following
problem is undecidable, just as in the usual ``invisible'' context free case
\cite{hunt, szymanski}:
given two VPDAs (or, equivalently, flattened TFAs) accepting languages $L_1$ and $L_2$
such that $L_1$ and $L_2$ are disjoint, is there a regular language $R$ such that
$R$ accepts all words from $L_1$, but no words from $L_2$?

A similar property
is also obtained for separating by other classes of languages, as long as the
corresponding problem for CFGs is undecidable, and the separating class has basic
closure properties and a pumping property -- the precise conditions are listed in the
sequel.
In \cite{hunt} it is shown that the separability problem of
context free languages is undecidable for any class which includes all
{\it definite} languages. On the other hand, 
it has been shown recently that the problem of separability of CFLs by
\emph{piecewise testable languages} is decidable \cite{wojtekczerwinski}.

Our method solves the following open problem, 
which has appeared on Rajeev Alur's website in early 2013 \cite{alur2013}:

\begin{quote}
A Challenging Open Problem

Consider the following decision problem: given two regular languages $L_1$ and $L_2$ of nested
words, does there exist a regular language $R$ of words over the tagged alphabet such that
Intersection($R$,$L_1$) equals $L_2$? [...]\end{quote}

We say that $L_2$ is a \emph{regular restriction} of $L_1$
iff the above holds.
Since disjoint languages $L_1$ and $L_2$ are separable iff $L_2$ is a regular
restriction of $L_1 \cup L_2$, and separability is undecidable,
restriction-regularity is undecidable too.

Rajeev Alur's question is inspired by \cite{streaming-pods}, where
it is shown that, for a fully recursive DTD, it is decidable whether
there is a regular language $R$ such that, for any valid XML document $w$,
it is decidable whether $w$ is valid with respect to $D$. This is a special case
of our result -- we take fully recursive DTDs instead of arbitrary finite automata
on trees, and we want to separate $L$ from its complement. It is stated as an open
problem in \cite{streaming-pods} whether the problem is still decidable for arbitrary
finite automata on trees.

On the other hand, in \cite{regvis} it is shown that it is decidable whether,
for a given visibly push-down language $L$, there exists a language $R$ such that
$R \cap F = L$, where $F$ is the language of all well matched words.

{\bf Acknowledgements.} Thanks to Charles Paperman for introducing me to
the separability problem for VPDAs, and for helping me with the references.

\section{Preliminaries}

We remind the basic notions of automata theory;  see \cite{hu79}.

For an alphabet $\Sigma$, $\Sigma^*$ denotes the set of words over $\Sigma$, and
$\epsilon$ denotes the empty word.

A \emph{deterministic finite automaton} (DFA) is a tuple $A = (\Sigma, Q, q_I, F, \delta$),
where $\Sigma$ is the alphabet of $A$, $Q$ is the set fo states, $q_I \in Q$ is the
initial state, $F \subseteq Q$ is the final state, and $\delta: Q \times \Sigma \ra Q$
is the transition function. We extend $\delta$ to $\delta: Q \times \Sigma^* \ra Q$
in the following way: $\delta(q,\epsilon) = q$, $\delta(q,wx) = \delta(\delta(q,w),x)$.

A \emph{context free grammar} (CFG) is a tuple $G = (V, \Sigma, R, S)$, where $V$ is the
set of non-terminal symbols, $\Sigma$ is the set of terminal symbols, $R$ is a set
of \emph{productions} of form $N \ra X_1 \ldots X_k$ where $N$ is a non-terminal symbol
and and each $X_i$ is either a terminal or non-terminal symbol, and $S \in V$ is
the start symbol. The language accepted by $G$, $L(G) \subseteq \Sigma^*$, is the set of
words which can be obtained from the start symbol $S$ by replacing non-terminal symbols
with words, according to the productions.

A \emph{binary flattened tree grammar} (BFG) over $\Sigma$ is a context free grammar, whose
set of terminal symbols is $\Sigma \cup \{\oa, \ca\}$, and every production is 
of form $N \ra t$, $N \ra \oa\ca$ or $N \ra \oa N_1 N_2 \ca$, where 
$N_1, \ldots, N_k$ are non-terminals, and $t$ is a terminal.
A BFG corresponds to the XML encoding of a 
regular language of trees ($\oa$ and $\ca$ correspond to the opening 
\verb:<a>: and closing \verb:</a>: tag, respectively), and languages recognized by
flattened tree grammars are visibly pushdown languages. These two facts are routine
to check -- we omit this to avoid having
to state the definitions of VPDAs and TFAs; we have decided to use flattened tree
grammars in this paper since they are easier to define than both of these formalisms.

\begin{definition}
We say that two languages $L_1$ and $L_2$ are \emph{separable} if there is a 
regular language $R$ such that for each $w \in L(G_1)$, $w \in R$,
but for each $w \in L(G_2)$, $w \notin R$.
\end{definition}

\begin{problem}[CFG-SEPARABILITY] \label{cflsep} {\ }

{\bf INPUT} Two context tree grammars $G_1$ and $G_2$ such that $L(G_1)$ and $L(G_2)$
are disjoint

{\bf OUTPUT} Are $L(G_1)$ and $L(G_2)$ separable?
\end{problem}

The CFL separation problem is known to be undecidable \cite{hunt, szymanski}. For
convenience, we include the idea of the proof here. Encode configurations of a 
Turing machine $M$ as words, and say that $w_1 \ra w_2$ iff a machine
in configuration $w_1$ reaches the configuration $w_2$ in next step. It can be
easily shown 
that (for simple encodings)
the languages $\{w_1 \# w_2^R: w_1 \ra w_2\}$ and $\{w_1^R \# w_2: w_1 \ra w_2\}$
are context free, and thus the languages $L_1$ and $L_2$ below are also context free.
They are separable iff $M$ terminates from the initial configuration $w_I$.
\begin{eqnarray*}
L_1 &=& \{w_1 \# w_2 \# \ldots w_2k \# a^{2k} : w_1 = w_I, w_{2i-1} \ra w_{2i}^R \} \\
L_2 &=& \{w_1 \# w_2 \# \ldots w_2k \# a^{k} : w_1 = w_I, w^R_{2i} \ra w_{2i+1} \}
\end{eqnarray*}

\begin{problem}[BFG-SEPARABILITY] \label{flattensep} {\ }

{\bf INPUT} Two BFGs $G_1$ and $G_2$ such that $L(G_1)$ and $L(G_2)$ are disjoint

{\bf OUTPUT} Are $L(G_1)$ and $L(G_2)$ separable?
\end{problem}

Our main result is the following:

\begin{theorem} \label{tflattensep}
The problem {\bf BFG-SEPARABILITY} is undecidable.
\end{theorem}

\section{Proof}

We will reduce {\bf CFG-SEPARABILITY} to {\bf BFG-SEPARABILITY}. To do this,
we will take two CFGs $G_1$ and $G_2$, and create two BFGs $G'_1$ and $G'_2$
such that $G'_1$ and $G'_2$ are separable iff $G_1$ and $G_2$ are.

Without loss of generality, we can assume that grammars $G_i$ do not accept the
empty word, or any word of length 1. We can also assume that these grammars
are in the \emph{Chomsky normal form}, that is, each production is of form
$N \ra N_1 N_2$ or $N \ra t$, where $N$, $N_1$ and $N_2$ are non-terminals,
and $t$ is a terminal. It is well known that any context free grammar is effectively
equivalent to a grammar in Chomsky normal form \cite{hu79}.

Given a context free grammar $G = (V, \Sigma, R, S)$
in Chomsky normal form, we will construct a 
binary flattened
tree grammar $G' = (V', \Sigma', R', S')$, in the following way:

\begin{itemize}
\item For each $X \in V$, we have a non-terminal $X'$. We also
have one special non-terminal $E'$. The starting symbol of $G'$ is $S'$.

\item For each production $N \ra t$ in $R$, we have the corresponding production in $R'$:
\begin{equation}
N' \ra t \label{gp_terminal}
\end{equation}

\item For each production $N \ra N_1 N_2$ in $R$, we have the corresponding
bracketed production in $R'$:
\begin{equation}
N' \ra \oa N_1' N_2' \ca \label{gp_bracketed}
\end{equation}

\item For each $X \in V$, we also have the following productions for $X'$:
\begin{eqnarray}
X' &\ra& \oa E' X' \ca \label{gp_leftempty} \\
X' &\ra& \oa X' E' \ca \label{gp_rightempty}
\end{eqnarray}

\item Where the productions for $E'$ are as follows:
\begin{eqnarray}
E' &\ra& \oa \ca \label{gp_leaf} \\
E' &\ra& \oa E' E' \label{gp_fork} \ca
\end{eqnarray}

\end{itemize}

Consider $\pi: \Sigma' \ra \Sigma$, the homomorphism which simply
removes the structural symbols $\oa$ and $\ca$. By applying
$\pi$ to all the production rules for $G'$, we obtain a grammar $\pi(G')$ which
accepts exactly $\pi(L(G))$. It is straightforward to check that $\pi(G')$ is
in fact equivalent to $G$ -- the only difference is that $E'$ is inserted in some
places, but all words generated by $E'$ reduce to the empty word after applying $\pi$.
Therefore, $\pi(L(G'_i))$ equals $L(G_i)$, which makes the following straightforward:

\begin{lemma}\label{easydir}
If $L(G_1)$ and $L(G_2)$ are separable, then so are $L(G_1')$ and $L(G_2')$.
\end{lemma}

\begin{proof}
$\pi^{-1}(R)$ is a regular language
which separates $L(G_1')$ and $L(G_2')$ -- in other words, the automaton separating
these two languages works exactly as the one separating $L(G_1)$ and $L(G_2)$
(it just ignores all the closing and structural symbols). \qed
\end{proof}

The rest of this section will prove the other direction:

\begin{theorem}\label{harddir}
If $L(G'_1)$ and $L(G'_2)$ are separable, then so are $L(G_1)$ and $L(G_2)$.
\end{theorem}

Assume that $L(G'_1)$ and $L(G'_2)$ are separable. Therefore, there is a finite automaton
$A$ such that $R = L(A)$ accepts all words from $L(G'_1)$, but no words from $L(G'_2)$. We say
that two words $w_1, w_2 \in \Sigma'^*$ are {\it syntactically equivalent} with respect to $R$
iff for any words $v, x \in \Sigma'^*$, we have $vwx\in R$ iff $vw'x \in R$.
Syntactic equivalence is a congruence with respect to concatenation.

\begin{lemma}\label{semigroup}
There is a number $\omega \in \bbN$ such that for any $w \in \Sigma'^*$,
$w^\omega$ is syntactically equivalent to $w^{2\omega}$ with respect to $R$.
\end{lemma}

\begin{proof}
The set $S$ of all the equivalence classes
is a semigroup with concatenation as the operation. This semigroup is called the
\emph{syntactic semigroup} of A, and it is finite -- if for 
two words $w_1$ and $w_2$ we have $\delta(q,w_1) = \delta(q,w_2)$ for each $q \in Q$,
then they are syntactically equivalent.
For any finite semigroup $(S, \cdot)$, there is a number $\omega \in \bbN$ such that 
for any $s \in S$, we have $s^{2\omega} = s^\omega$ -- since $k \mapsto s^k$
yields an ultimately cyclic sequence with period at most $|S|$, 
$\omega = |S|!$ will work. \qed
\end{proof}

We say that $T: \Sigma^* \ra \Sigma'^*$ is a \emph{padding} iff
there exist $e_L, e, e_R \in (\Sigma'-\Sigma)^*$ such that,
for any $w=t_1 \ldots t_n \in \Sigma^*$,
$T(w)$ is the word $e_L t_1 e t_2 e \ldots e t_n e_R$.

\begin{lemma}\label{translemma}
For a padding $T$, 
the languages $L(G_1)$ and $L(G_2)$ are separable, iff $T(L(G_1))$ and $T(L(G_2))$ are.
\end{lemma}

\begin{proof}
The forward direction is straightforward, and proven just as Lemma \ref{easydir}
-- the automaton simply ignores all the symbols from $\Sigma'-\Sigma$ (in fact,
the stronger version of this lemma where $e_L, e, e_R \in \Sigma'^*$ is also true).

For the backward direction, we take the DFA $A' = (\Sigma', Q, q_I, F, \delta)$.
We can assume that there are no transitions to the initial state $q_I$ in $A'$ --
otherwise, we create a copy of $q_I$ and make it the new initial state.

We construct a new DFA $A'' = (\Sigma, Q, q_I, F', \delta')$ in the following way: take
$\delta'(q_I,t) = \delta(q_I, e_L t)$, and $\delta'(q,t) = \delta(q, et)$ for $q \neq q_I$.
For $F'$ we take the set of states $q$ such that $\delta(q,e_R) \in F$. 
The
automaton $A''$ working on $w \in \Sigma^*$ simulates the automaton $A'$ working
on $T(w)$, hence it accepts $w$ iff $A'$ accepts $T(w)$. \qed
\end{proof}

\begin{lemma}\label{pumplemma}
There is a padding $T$ with the following property:
for each $w \in L(G_i)$, there is a word $w' \in L(G'_i)$
which is equivalent to $T(w)$ with respect to $R$.
\end{lemma}

This proves Theorem \ref{harddir} and thus Theorem \ref{tflattensep}. Indeed,
we will show that $T(L(G_1))$ and $T(L(G_2))$ are separated by $R$ -- then, after
applying Lemma \ref{translemma}, we get our claim.

Consider $w \in L(G_i)$; we have to show that $L(A)$ accepts $T(w)$ iff $i=1$.
From Lemma \ref{pumplemma} we know that there is some $w' \in L(G'_i)$ which
is equivalent to $T(w)$ with respect to $R$. Therefore, 
we know that $T(w) \in R$ iff
$w' \in R$, and since $w' \in L(G'_i)$, $T(w) \in R$ iff $i=1$. \qed

\begin{proof}[of Lemma \ref{pumplemma}]

\def\ob{\heartsuit}
\def\cb{{\overbar{\heartsuit}}}
Let $\ob = \oa \oa \ca$, $\cb = \oa \ca \ca$. Thus, for any 
non-terminal $N' \in V'$, by applying one of productions (
\ref{gp_leftempty}, \ref{gp_rightempty} or \ref{gp_fork}) and then 
(\ref{gp_leaf}), we have $N' \ra^* \ob N' \ca$ and 
$N' \ra^* \oa N' \cb$. Also, let $\nu = \omega-1$.

By applying the above many times to terminals $K'$ and $L'$, we get
$K' \ra^* b_1 K' b_2$ and $C' \ra^* c_1 L' c_2$, where:

\begin{eqnarray}
b_1 &=& (\oa^\nu \ob^\nu)^\omega \oa^\nu \\
b_2 &=& \cb^\nu (\ca^\nu \cb^\nu)^\omega \\
c_1 &=& \ob^\omega (\ob^\nu \oa^\nu)^\omega \\
c_2 &=& (\cb^\nu \ca^\nu)^\omega \ca^\nu
\end{eqnarray}

Now, whenever we have a production $N \ra KL$ in $R$, we can do the following in $R'$:
\begin{equation} \label{goodjob}
N' \ra \oa K'L' \ca \ra^* \oa b_1 K' b_2 c_1 L' c_2 \ca
\end{equation}

This can be written as $N' \ra e_L K' e L' e_R$, where:

\begin{eqnarray}
e_L &=& \oa b_1 \\
e   &=& b_2 c_1 \\
e_R &=& c_2 \ca
\end{eqnarray}

We claim that the padding $T$ given by the words $e_L$, $e$, $e_R$ defined above
satisfies our claim.

Indeed, take $w \in L(G_i)$. Consider the derivation tree of $w$ in $G_i$;
repeat this derivation in $G'_i$, replacing each production $N \ra KL$ with
$N' \ra e_L K' e L' e_R$ according to
the chain of productions (\ref{goodjob}) above, and each production $N \ra t$ with $N' \ra t$
(\ref{gp_terminal}).
In the end, for $w = t_1 \ldots t_n$, we obtain the word $w' \in L(G'_i)$, which
contains the symbols $t_1, \ldots, t_n$ separated with $e$, possibly accompanied
by $e_L$'s on the right side and $e_R$'s on the left side, and with at least one
$e_L$ before $t_1$ and at least one $e_R$ after $t_N$. In other words,

\[
w' = e_L^{l_1} t_1 e_R^{r_1} e e_L^{l_2} t_2 e_R^{r_2} e e_L^{l_3} \ldots
t_n e_R^{r_N} \]

where all $l_i$, $r_i$ are integers, and $l_1,r_n \geq 1$.
Remembering that $\oa$ is a left factor of $\ob$ and thus $\oa^{\omega+\nu} \ob \equiv 
\oa^\nu \ob$, and similarly $\cb \ca^{\omega+\nu} \equiv \cb \ca^\nu$,
it can be checked that the following equivalences hold:

\begin{itemize}
\item $e_L e_L$ is equivalent to $e_L$:
\begin{eqnarray*}
&& e_L e_L = \oa b_1 \oa b_1 = \\ &=&
\oa (\oa^\nu \ob^\nu)^\omega \oa^\nu \oa (\oa^\nu \ob^\nu)^\omega \oa^\nu \equiv \\
&\equiv& \oa (\oa^\nu \ob^\nu)^\omega \oa^\omega (\oa^\nu \ob^\nu)^\omega \oa^\nu \equiv \\
&\equiv& \oa (\oa^\nu \ob^\nu)^\omega (\oa^\nu \ob^\nu)^\omega \oa^\nu \equiv
\oa (\oa^\nu \ob^\nu)^\omega \oa^\nu = \oa b_1 = e_L
\end{eqnarray*}

\item $e_R e_R$ is equivalent to $e_R$:
\begin{eqnarray*}
&& e_R e_R = c_2 \ca c_2 \ca = \\ &=&
(\cb^\nu \ca^\nu)^\omega \ca^\nu \ca (\cb^\nu \ca^\nu)^\omega \ca^\nu \ca \equiv \\
&\equiv& (\cb^\nu \ca^\nu)^\omega \ca^\omega (\cb^\nu \ca^\nu)^\omega \ca^\nu \ca \equiv \\
&\equiv& (\cb^\nu \ca^\nu)^\omega (\cb^\nu \ca^\nu)^\omega \ca^\nu \ca \equiv
(\cb^\nu \ca^\nu)^\omega \ca^\nu \ca = c_2 \ca = e_R
\end{eqnarray*}

\item $e e_L$ is equivalent to $e$:
\begin{eqnarray*}
&& e e_L = b_2 c_1 \oa b_1 = \\ &=&
b_2 \ob^\omega (\ob^\nu \oa^\nu)^\omega \oa (\oa^\nu \ob^\nu)^\omega \oa^\nu \equiv \\
&\equiv& b_2 \ob^\omega (\ob^\nu \oa^\nu)^\omega \oa \oa^\nu (\ob^\nu \oa^\nu)^\omega \equiv \\
&\equiv& b_2 \ob^\omega (\ob^\nu \oa^\nu)^\omega (\ob^\nu \oa^\nu)^\omega \equiv
b_2 \ob^\omega (\ob^\nu \oa^\nu)^\omega = b_2 c_1 = e
\end{eqnarray*}

\item $e_R e$ is equivalent to $e$:
\begin{eqnarray*}
&& e_R e = c_2 \ca b_2 c_1 = \\ &=&
(\cb^\nu \ca^\nu)^\omega \ca^\nu \ca \cb^\nu (\ca^\nu \cb^\nu)^\omega c_1 \equiv \\
&\equiv& (\cb^\nu \ca^\nu)^\omega \cb^\nu (\ca^\nu \cb^\nu)^\omega c_1 \equiv \\
&\equiv& (\cb^\nu \ca^\nu)^\omega (\cb^\nu \ca^\nu)^\omega \cb^\nu c_1 
\equiv (\cb^\nu \ca^\nu)^\omega \cb^\nu c_1 = b_2 c_1 = e
\end{eqnarray*}
\end{itemize}

These rules allow us to reduce all $l_i$ and $r_i$ to zeros (except $l_1$ and $r_n$,
which can be reduced to 1). Hence, the word $w'$ is equivalent to $T(w)$. \qed
\end{proof}

\section{Why binary trees?}

Our proof uses {\it binary} flattened tree grammars -- that is,
productions of form $N \ra \oa N_1 N_2 \ldots N_k \ca$ are allowed only for $k=0$
or $k=2$.

If we allowed such rules for any $k$, a somewhat easier proof would work.
The general structure is the same, but
it is not required to have $G_i$'s in the Chomsky normal form, 
and the grammar $G'$ allows inserting empty brackets with productions 
$E' \ra \oa\ca | \oa E' \ca | \oa E'E' \ca$ before and after every symbol:
$N \ra \oa E' X_1 E' X_2 \ldots X_k E' \ca$, and $e_L=e=e_R=(\oa^\omega \ca\omega)^\omega$
can be obtained by pumping each pair of $\oa\ca$ $\omega$ times, and then
using the rule $E' \ra E'E'$ to make sure that $\oa^\omega \ca^\omega$ appears
$k\omega$ times between each two consecutive terminals.

While such a proof 
was sufficient to solve the problem for visibly pushdown languages, and for
XML encoding of trees, it left us with a craving for more, for the following
reasons.

First, if we consider languages of terms, it is natural to consider
different symbols for nodes with different arities (numbers of children) --
while the simpler proof above heavily uses the fact that the XML encoding cannot
tell whether $\oa$ comes from the structurally significant rule $N \ra \oa E' X_1 E' X_2 \ldots X_k E' \ca$,
or it is a fake inserted by the $X \ra \oa X \ca$ or $E' \ra \oa E' E' \ca$ rules.
Since the arities here are respectively $2k+1,$ 1 and 2, this fails if
the automaton sees them.

Moreover, if we consider the process of flattening a tree as the result of an
automaton which traverses the tree recursively and react to what it is seeing on its path,
it is natural to assume that such an automaton sees the arity, and moreover, 
between returning from the $i$-th child of $v$ and progressing to the $(i+1)$-th child,
the automaton should see that $i$ of $n$ children are progressed. This corresponds
to flattening rules of form $N \ra \oa_0^k X_1 \oa_1^k X_2 \oa_2^k X_3 \ldots X_k \oa_k^k$.

Restricting ourselves to the binary case allows us to solve the problem in full
generality -- while the encoding in the definition of BFG does not explicitely 
say whether $\oa$ comes from 
the ``empty leaf'' rule $E' \ra \oa\ca$ or from one of the binary branching rules,
a finite automaton can easily tell which one is the case by looking at the neighborhood.
Also, while we do not write the infix structural symbols $\oa_1^2$ explicitely,
the automaton can tell that it is at such a branching point iff the last symbol was $\ca$,
and the next one is $\oa$.

In the binary case, a computer program was used to find the 
appropriate words $e,e_L,e_R$ which work in Lemma \ref{pumplemma}.

\section{Conclusion}
We can also consider the separation problem for other classes of languages --
that is, is it decidable whether languages accepted by two BFGs are separable
by a language of class $\calC$? From the proof above, this problem is undecidable
for class $\calC$, as long as the following conditions are satisfied:

\begin{itemize}
\item The respective problem for CFGs is undecidable.
\item The class $\calC$ has the following pumping property: for any $L \in \calC$,
there exists some $\omega$ such that for any words $v$, $w$, $x$,
$v{w^\omega}x \in L$ iff $v{w^{2\omega}}x \in L$. (This is Lemma 
\ref{semigroup}, and it is used in the proof of Lemma \ref{pumplemma}.)
\item If $\pi$ is a homomorphism which ignores a subset of symbols, 
and $L \in \calC$, then $\pi^{-1}(L) \in \calC$. (Used in the proofs of Lemma 
\ref{easydir} and \ref{translemma}.)
\item If $T(t_1 \ldots t_k) = e_L t_1 e t_2 e \ldots e t_k e_R$
for some $s$, and $L \in \calC$, then $T(L) \in \calC$. (Used in the proof of \ref{translemma}.)
\end{itemize}

Hence, the separability problem is also undecidable for other classes of languages,
such as the class of languages recognizable by first-order logic.

\bibliographystyle{plain}
\bibliography{visibly}

\end{document}